\documentclass[AMA,STIX1COL]{WileyNJD-v2}

\articletype{Research Article}%
\usepackage{upgreek}

\received{xx}
\revised{xx}
\accepted{xx}

\begin{document}

\title{Bayesian response adaptive randomization design with a composite endpoint of mortality and morbidity}

\author[1]{Zhongying Xu*}

\author[1]{Andriy I. Bandos}

\author[2]{Tianzhou Ma}

\author[1]{Lu Tang}

\author[3]{Victor B. Talisa}

\author[1,4]{Chung-Chou H. Chang}

\authormark{AUTHOR ONE \textsc{et al}}

\address[1]{Department of Biostatistics, School of Public Health, University of Pittsburgh, Pittsburgh, PA, USA}

\address[2]{Department of Epidemiology and Biostatistics, School of Public Health, University of Maryland, College Park, MD, USA}

\address[3]{Department of Critical Care Medicine, School of Medicine, University of Pittsburgh, Pittsburgh, PA, USA }

\address[4]{Department of Medicine, School of Medicine, University of Pittsburgh, Pittsburgh, PA, USA}

\corres{*Chung-Chou H. Chang, University of Pittsburgh,
200 Meyran Avenue,
Suite 200,
Pittsburgh, PA, USA. \email{changj@pitt.edu}}

\abstract[Summary]{Allocating patients to treatment arms during a trial based on the observed responses accumulated prior to the decision point, and sequential adaptation of this allocation, could minimize the expected number of failures or maximize total benefit to patients. In this study, we developed a Bayesian response adaptive randomization (RAR) design targeting the endpoint of organ support-free days (OSFD) for patients admitted to the intensive care units (ICU). The OSFD is a mixture of mortality and morbidity assessed by the number of days of free of organ support within a predetermined time-window post-randomization. In the past, researchers treated OSFD as an ordinal outcome variable where the lowest category is death. We propose a novel RAR design for a composite endpoint of mortality and morbidity, e.g., OSFD, by using a Bayesian mixture model with a Markov chain Monte Carlo sampling to estimate the posterior probability distribution of OSFD and determine treatment allocation ratios at each interim. Simulations were conducted to compare the performance of our proposed design under various randomization rules and different alpha spending functions. The results show that our RAR design using Bayesian inference allocated more patients to the better performing arm(s) compared to other existing adaptive rules while assuring adequate power and type I error rate control across a range of plausible clinical scenarios.}

\keywords{Bayesian, Composite endpoint, Organ support free days, Response adaptive randomization}

\maketitle

\section{Introduction}\label{sec:intro}

To balance between individual ethics and statistical efficiency, most randomized clinical trials use fixed treatment allocation ratios in which patients are assigned among the comparative arms with fixed probability throughout the study \citep{lin2016general}. Though such fixed treatment allocation procedure remains to be the most used allocation rules in practice, more flexible and efficient allocation rules have become increasingly popular. In the response-adaptive randomization (RAR) design, treatment allocation probabilities for each arm are adjusted at different stages of the trial based on the accumulated information of treatment assignments and responses of patients already in the trial. RAR aims to achieve the objective of either maximizing patients' total benefits by allocating more patients to receive the superior treatment or minimizing the expected number of failures for a targeted power \citep{rosenberger2001optimal}. To date, numerous methods for RAR with a single type of endpoint had been developed \citep{wei1978randomized,ivanova2003play,villar2015multi,rosenberger2001optimal,hu2004asymptotic,eisele1995central}. Practical applications of RAR have captured much attention since the COVID-19 pandemic started \citep{angus2020optimizing,magaret2020multigroup,villar2021temptation}.

Instead of using a single type of response, many studies used composite clinical responses as the trial endpoints \citep{anand2009cardiac,cannon2004intensive,shepherd1995prevention}. In some cases, a combination of several relevant events was used to increase the combined event rate and reduce the sample size needed. However, if the components are not of equal clinical importance, i.e., one event is much less severe than the other, the interpretation of such composite response may be unclear \citep{ferreira2007problems,tomlinson2010composite}. To overcome this issue, studies have analyzed each outcome separately and separated the composite response into a primary outcome and a leading secondary outcome. Other studies treated this type of outcome as an ordinal variable and made estimation and inference accordingly. For example, if the outcome of interest was organ support free days (OSFD) for critically ill patients in the intensive care units (ICUs). 
For each of the patient who was discharged alive from ICU, his or her OSFD was calculated up to the end of the study \citep{abdelhady2021effect}. If a patient died during the ICU stay, his or her OSFD was assigned an arbitrary low number (e.g., –1 day) to indicate the worst outcome. 

In spite of the increasing development and implementation of RAR procedures in clinical trials, an RAR procedure that targets a composite endpoint of mortality and morbidity is still lacking. In this study, we propose a novel patient-benefit-oriented Bayesian response-adaptive randomization (BRAR) scheme and apply it to a multicenter, unblinded, phase II or III design with the primary endpoint of a composite endpoint of mortality and morbidity (e.g., OSFD by day 28) among sepsis patients admitted to the ICUs. To handle the composite endpoint, we introduce a Bayesian mixture model to handle the composite endpoint in the RAR design. Moreover, we propose three adaptive allocation rules under the BRAR scheme. Our method also adjusts the stopping boundaries during the course of the trial to control the overall type I error rate.

In Section 2 to Section 5, we illustrate the Bayesian mixture model for RAR with a composite endpoint, and demonstrate the asymptotic properties of the parameter estimators. Simulation study and the results are shown in Section 6. We summarize the work in Section 7.

\section{Model and Parameter Estimate}\label{sec:model}

\subsection{Notation and outcome transformation}
\label{subsec:notation}
Consider an RAR trial with a total of $N$ patients and $J$ decision stages. Patients will be sequentially enrolled in the trial and randomly assigned to receive one of $K$ competing treatments. At the initial stage, roughly equal number of patients will be allocated to each of the treatment arms.  At each of the subsequent decision stages, treatment allocation probabilities will be adapted, and the calculations will be based on the accumulated information on treatment assignments and the responses from patients in the previous stages. Let $n_{jk}$ be the number of patients who are assigned to arm $k$ $(k=0,1,\dots,K)$ at stage $j$ $(j=1,\dots,J)$, where $k=0$ denotes the control arm. Without loss of generality, we assume an equal number of patients distributed to each stage, therefore, each stage consists of $n_j=\sum\nolimits_{k=0}^K n_{jk}=N/J$ patients. 

For each sepsis patient admitted to the ICU, we observe two types of responses: whether the patient dies during the ICU stay and how many days the patient is free of organ support if he or she is discharged alive. Denote by $\tau_{ik}$ the death indicator which equals to 1 if patient $i$ assigned to arm $k$ died during the ICU stay, and 0 otherwise. Denote by $Y_{ik}$ the number of days patient $i$ assigned to arm $k$ was on organ support (a.k.a., the OSFD). If $\tau_{ik}=1$, no OSFD is observed for this patient thus $Y_{ik}$ is undefined. Note that a larger value of OSFD indicates a better prognosis. 

Considering that OSFD is measured by Day 28 after ICU admission, the range of observed $Y_{ik}$ falls in $[1,28]$, with OSFD greater than 28 being censored at 28. Hence, when the observed $Y_{ik}$ is $28$, it corresponds to either the true OSFD is equal to 28 or is greater than 28 but censored at 28. Figure \ref{fig1}(A) depicts the distribution of OSFD on the original scale $Y_{ik}$ from a simulated data set. The overall distribution of $Y_{ik}$ is skewed to the right and has a peak at Day 28 with a large proportion of censoring. $Y_{ik}$ is undefined for death events with $\tau_{ik}=1$.

For modeling convenience, we define a random variable of composite response $D_{ik} = \big\{\lceil \log(30) \rceil -\log(30-Y_{ik})\big\}\times (1-\tau_{ik})$, where $\lceil \log(30) \rceil$ is the least integer greater than $\log(30)$. If a patient died by Day 28, $D_{ik}=0$ by definition. If a patient is alive, $D_{ik}=\lceil \log(30) \rceil-\log(30-Y_{ik})$ is a monotonic increasing function of $Y_{ik}$. The composite-response random variable $D_{ik}=\lceil \log(30) \rceil-\log(2)$ includes those patients with true OSFD of 28 and those with censored OSFD by Day 28. Compared to $ Y_{ik}$, $D_{ik}$ (for $\tau_{ik}=0$) is much more normally distributed with a peak at $\{\lceil \log(30) \rceil-\log(2)\}$ indicating the censoring. See Figure \ref{fig1}(B) for an illustration of the impact of this transformation on the OSFD outcome. 

We denote the observed data collected up to the $(j-1)^{th}$ stage as $\mathfrak{D}_{{j-1}_+}=\{(D_{ik}, \tau_{ik}): j=2, \ldots, 10 ; i=1,\ldots,n_{{j-1}_+k}; k=0,\ldots,K\}$, where $n_{{j-1}_+k}$ denotes the number of patients who were assigned to arm $k$ up to stage ($j-1$).

\begin{figure}
\centering
\includegraphics[scale=0.6]{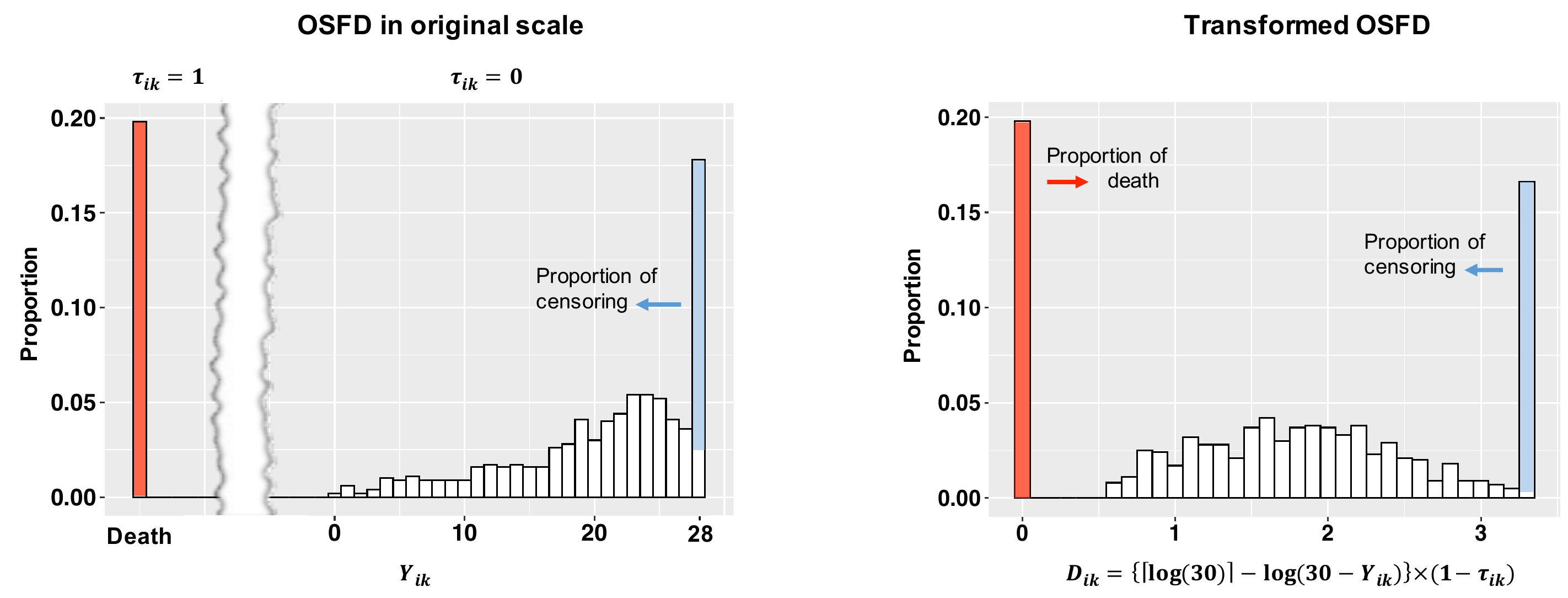}
\caption{(A) OSFD in original scale. (B) Transformed OSFD.}
\label{fig1}
\end{figure}

\subsection{Bayesian mixture model and prior setting}
\label{subsec:model}

Figure \ref{fig2} depicts the conceptual full hierarchical Bayesian model that can be used to model the composite endpoint of mortality and morbidity. Among patients who were discharged alive at the $k$th ($0\le k\le K$) arm, we consider the following probability density function of $D_{ik}$ in our proposed two-component mixture model:
\begin{equation}
\label{method:eq1}
\begin{split} 
 f(D_{ik}|\tau_{ik}=0) & = \omega_k + (1-\omega_k) TN\{D_{ik};\mu_k,\sigma^2_k, \lceil \log(30) \rceil-\log(30),\lceil \log(30) \rceil-\log(2)\}; \\
 \tau_{ik} & \sim Ber(\lambda_k),
 \end{split}
\end{equation} 
where $TN\{\mu_k,\sigma^2_k,\lceil \log(30) \rceil-\log(30),\lceil \log(30) \rceil-\log(2)\}$ represents the probability density function of truncated normal distribution with mean $\mu_k$, variance $\sigma^2_k$, and truncated at $\{\lceil \log(30) \rceil-\log(30)\}$ and $\{\lceil \log(30) \rceil-\log(2)\}$. By definition, $D_{ik}=0$ when $\tau_{ik}=1$. $\omega_k$ represents the proportion of censoring peaking at $D_{ik}=\lceil \log(30) \rceil-\log(2)$. Associated with $D_{ik}$, we have a latent binary indicator variable $Z_{ik}$, which is equal to 1 when $D_{ik}$ is obtained from the censored component and 0 otherwise. By definition from the mixture model, we have $P(Z_{ik}=1)=\omega_{k}$. 

We assume the following noninformative conjugate priors for the parameters in (\ref{method:eq1}): 
\begin{equation}\label{method:eq2}
\begin{split}
& \mu_{k} \sim N(0, 10^4), \hspace{1cm} \sigma^2_{k} \sim IG(10^{-4}, 10^{-4}), \\
& \omega_{k} \sim Unif(0,1), \hspace{1cm}  \lambda_{k} \sim Unif(0,1) \\
\end{split}
\end{equation}

The full conditional distribution of all parameters in the model has analytic forms (see Appendix) and we use the Gibbs sampling algorithm \citep{geman1984stochastic} to generate samples from the posterior distributions for further inference.

\subsection{The Estimand}
\label{subsec:estimand}
To compare treatment effects across arms, we first estimate $\theta_{k}  = E(D_{ik})$, the mean transformed OSFD for arm $k$. From (\ref{method:eq1}), we have:

\begin{equation}
\label{method:eq3}
\begin{split} 
\theta_k=  E(D_{ik}) & = E(D_{ik}|\tau_{ik}=0)P(\tau_{ik}=0)  + E(D_k|\tau_{ik}=1)P(\tau_{ik}=1) \\ 
          &  = E(D_{ik}|\tau_{ik}=0)(1-\lambda_k) + 0 * \lambda_k \\   
          &  = \big[ \omega_k \big\{\lceil \log(30) \rceil-\log(2)\big\} + (1-\omega_k)( \mu_k - \sigma_k  R_{TN,k} )\big] (1-\lambda_k) , \\ 
 \end{split}
\end{equation}
where $R_{TN,k}= \frac{  \phi\{\frac{\lceil \log(30) \rceil-\log(2)-\mu_k}{\sigma_k}\}  - \phi\{\frac{\lceil \log(30) \rceil-\log(30)-\mu_k}{\sigma_k}\} } { \Phi\{\frac{\lceil \log(30) \rceil-\log(2)-\mu_k}{\sigma_k}\}  - \Phi\{\frac{\lceil \log(30) \rceil-\log(30)-\mu_k}{\sigma_k}\}  } $, $\phi(.)$ and $\Phi(.)$ stand for the probability density function and the cumulative distribution function of the standard normal distribution, respectively. Derivation of $Var(D_{ik})$, the variance of the OSFD for arm $k$ can be found in the Appendix. We can estimate $\theta_k$ using its posterior mean:

\begin{equation}
\label{method:eq4}
\hat{\theta}_k = \sum\limits_{b=1}^B \theta^{(b)}_k/ B , 
\end{equation} 
where $B$ is the total number of MCMC iterations after the burn-in period. Based on the parameter estimates at the $b$th iteration and equation (\ref{method:eq3}), $\theta^{(b)}_k$ is the posterior estimate of $\theta_k$.
  
We define our estimand $\xi_k = \theta_k - \theta_0$ as the difference in expected value of $D_{ik}$ between the treatment arm $k$ and the control arm ($k=0$). The estimand can be estimated by taking the posterior mean of the differences for all possible pairs: 
\begin{equation}
\label{method:eq5}
\hat{\xi}_k= \sum\limits_{b=1}^B \sum\limits_{b'=1}^B (\theta^{(b)}_k - \theta^{(b')}_0 ) / B^2, 
\end{equation} 
where $\theta^{(b)}_k$ and $\theta^{(b')}_0$ are the posterior estimates of $\theta_k$ and $\theta_0$ at the $b^{th}$ and $(b')^{th}$ iterations, respectively.

\begin{figure}
\centering
\includegraphics[scale=0.5]{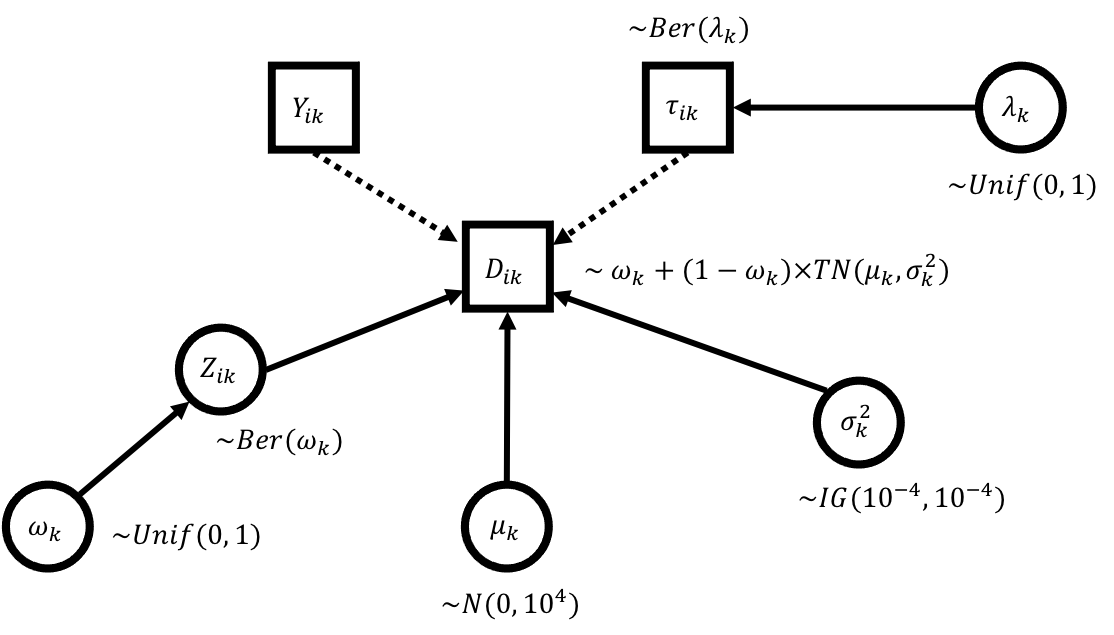}
\caption{Full Bayesian hierarchical model.}
\label{fig2}
\end{figure}

\subsection{Hypothesis setting}
\label{subsec:hypothesis}
For a study of $K$ active treatment arms, we consider the null hypothesis of:
  
$$ H_0: \theta_{k} - \theta_{0} \le \delta, \text{  } \forall{k}\in\{1,2, \dots,K\}, $$
where $\delta$ is a clinically meaningful threshold chosen for the proposed trial of OSFD. We consider the following three alternative hypotheses: 

$$ H_{a1}: \sum\limits_{k=1}^K I( \theta_{k} - \theta_{0} >\delta)\ge 1; $$
$$ H_{a2}:  \sum\limits_{k=1}^K I( \theta_{k} - \theta_{0} >\delta)\ge r, \text{ where } 2\le r \le K-1; $$
$$ H_{a3}: \sum\limits_{k=1}^K I( \theta_{k} - \theta_{0} >\delta)= K,$$
where $I(\cdot)$ is the indicator function. The first alternative hypothesis $H_{a1}$ aims to detect at least one best arm; the second alternative hypothesis $H_{a2}$ aims to detect treatment effects in at least $r$ well-performing arms; and the third hypothesis $H_{a3}$ aims to detect treatment effect in all arms. We evaluate the posterior probability of superiority for each of the three alternative hypotheses in the simulation studies.

\section{Asymptotic properties of the parameter estimate}
\label{subsec:asymptotic}
In this section, we show the asymptotic properties of the posterior mean estimates of $\hat{\theta}_k$ and $\hat{\xi}_k$.

\begin{theorem}
\label{thm1}
(Consistency). Suppose $D_{1k},\ldots,D_{nk}$ are drawn i.i.d. from the correctly specified model in Equation \eqref{method:eq1} $\{p(D_k|\Delta_k),\Delta_k \in \Uptheta \}$ with 
$\mathbf{P}_{\Delta_k^*}$ for some true parameter $\Delta_k^*=(\lambda_k^*,\omega_k^*,\mu_k^*,\sigma^{2*}_k)$. The true parameter $\Delta_k^*$ is drawn from the prior $\mathbf{\pi}(\Delta)$ and $\theta^*_k=g(\Delta^*_k)$ where $g(.)$ is a continuous function specified in Equation \eqref{method:eq3}, then the posterior mean estimate $\hat{\theta}_k \overset{p}{\to} \theta^*_k$. 
\end{theorem}

\begin{proof}
The parameter in the model $\Delta_k=(\lambda_k,\omega_k,\mu_k,\sigma^2_k)$ is of finite dimension, by Lemma 1, the Doob's theorem (see Appendix), the posterior mean estimate $\hat{\Delta}_k$ is a consistent estimator, i.e., $\hat{\Delta}_k \overset{p}{\to} \Delta^*_k $. Since $\hat{\theta}_k= g(\hat{\Delta}_k) $ is a continuous function of $\hat{\Delta}_k$, by the continuous mapping theorem, we have $\hat{\theta}_k \overset{p}{\to} \theta^*_k$.   
\end{proof}

\begin{proposition}
\label{prop1}
Suppose $D_{1k},\ldots,D_{nk}$ are drawn i.i.d. from the correctly specified model in Equation \eqref{method:eq1} $\{p(D_k |\Delta_k),\Delta_k \in \Uptheta \}$ with 
$\mathbf{P}_{\Delta_k^*}$ for some true parameter $\Delta_k^*=(\lambda_k^*,\omega_k^*,\mu_k^*,\sigma^{2*}_k)$ for $0\le k \le K$, and assuming independence between different arms. The true parameter $\Delta_k^*$ is drawn from the prior $\mathbf{\pi}(\Delta)$  and $\theta^*_k=g(\Delta^*_k)$ and $\xi^*_k = \theta^*_k - \theta^*_0 $ for $1\le k \le K$, then the posterior mean estimate $\hat{\xi} \overset{p}{\to} \theta^*_k - \theta^*_0 $. 
\end{proposition}

\noindent This proposition immediately follows from Theorem \ref{thm1} when the different arms are independent. We do not repeat the proof here. 

\begin{theorem}
\label{thm2}
(Asymptotic Normality). For a correctly specified model in Equation \eqref{method:eq1} $\{p(D_k|\Delta_k),\Delta_k \in \Uptheta \}$ with 
$\mathbf{P}_{\Delta_k^*}$ under certain regularity conditions and with Fisher's information matrix $\mathbf{I}_{\Delta^*}$. $\Delta_k^*$ is drawn from the prior $\mathbf{\pi}(\Delta)$ and $\theta^*_k=g(\Delta^*_k)$ where $g(.)$ is a continuous function specified in Equation \eqref{method:eq3}, then the posterior mean estimate follows an asymptotically normal distribution: $\sqrt{n} (\hat{\theta}_k - \theta^*_k) \overset{d}{\to} N\{0, \nabla \theta_k(\Delta_k)^T \mathbf{I}^{-1}_{\Delta^*}  \nabla \theta_k(\Delta_k) \} $, where $\nabla \theta_k(\Delta_k)$ is the gradient of $\theta_k$ with respect to $\Delta_k$.  
\end{theorem}

\noindent We show a proof of this theorem and a derivation of the asymptotic variance in the Appendix. 

\begin{proposition}
For a correctly specified model in Equation \eqref{method:eq1} $\{p(D_k|\Delta_k),\Delta_k \in \Uptheta \}$ with 
$\mathbf{P}_{\Delta_k^*}$ under certain regularity conditions and with Fisher's information matrix $\mathbf{I}_{\Delta^*_k}$ for $0\le k \le K$ and assuming independence between different arms. $\Delta_k^*$ is drawn from the prior $\mathbf{\pi}(\Delta)$ and $\theta^*_k=g(\Delta^*_k)$ and $\xi^*_k = \theta^*_k - \theta^*_0 $ for $1\le k \le K$, then the posterior mean estimate follows an asymptotically normal distribution: $\sqrt{n} (\hat{\xi}_k - \xi^*_k) \overset{d}{\to} N\big\{0, \nabla \theta_k(\Delta_k)^T \mathbf{I}^{-1}_{\Delta_k^*}  \nabla \theta_k(\Delta_k) + \nabla \theta_0(\Delta_0)^T \mathbf{I}^{-1}_{\Delta_0^*}  \nabla \theta_k(\Delta_k) \big\} $, where $\nabla \theta_k(\Delta_k)$ and $\nabla \theta_0(\Delta_0)$ are the gradients of $\theta_k$ and $\theta_0$ with respect to $\Delta_k$ and $\Delta_0$, respectively. 
\end{proposition}

\noindent This proposition immediately follows from Theorem \ref{thm2} when the different arms are independent. We also provide a derivation of the asymptotic variance in the Appendix. 

\section{Interim analysis and adaptation rules}\label{sec:rules}
The purpose of our Bayesian adaptive design is to identify the best performing treatment arm and allocate more patients to this arm. To reward the well-performing arm(s) and remove poor-performing arm(s), we consider three types of adaptation at each interim analysis stage: (a) re-estimation of the treatment allocation ratios, (b) arm suspension due to poor performance, and (b) early trial stopping due to efficacy. In the following subsections, we describe how to implement each of these adaptations in details.

\subsection{Re-estimation of the treatment allocation ratios}\label{subsec:rar}
RAR is applied here so more patients can be assigned the better treatment arm. At each interim stage $j$, the allocation probabilities for all treatment arms ($1\le k\le K$) are updated to be proportional to the posterior probability that the arm is the best arm: 

\begin{equation}
\label{rule:eq1}
 p_{k,j} = P(\theta_{k} = \max_{1\le h\le K} \theta_{h} | \mathfrak{D}_{j-1_+} )  = P(\bigcap\limits_{h\neq k, 1\le h\le K} \theta_{k} > \theta_{h}  | \mathfrak{D}_{j-1_+}), 
\end{equation}
where $\mathfrak{D}_{j-1_+}$ is the data collected up to ${(j-1)}^{th}$ stage. This posterior probability can be estimated based on the posterior samples acquired from the Bayesian model on $\mathfrak{D}_{j-1_+}$. 

With the estimated $\hat{p}_{k,j}$, we propose the following three RAR rules depending on the proportion of patients assigned to the control arm: \\

\noindent\textbf{RAR Rule I:}
\begin{enumerate}
\item Fix the allocation proportion for the control arm to $1/(K+1)$. 
\item For each treatment arm $k$ ($1\le k\le K$), compute $\hat{p}_{k,j}$.
\item Scale $\hat{p}_{k,j}$ so $\sum\limits_{k=1}^K \hat{p}_{k,j}=1$. Allocate the remaining $K/(K+1)$ of samples to each treatment arm proportional to the corresponding scaled probability. 
\end{enumerate}

\noindent\textbf{RAR Rule II:}
\begin{enumerate}
 \item Compute the scaled $\hat{p}_{k,j}$ as in Rule I. 
 \item Allocate to the best arm first, with the maximum allocation probability bounded by 0.8 to avoid extremely low numbers in the control arm.
\item For the remaining samples, allocate to the control arm with the average proportion, i.e., allocate $1/K$ of the remaining samples to the control arm.    
\item For the remaining samples, allocate to the other treatment arms with the corresponding scaled probabilities. 
\end{enumerate}

\noindent\textbf{RAR Rule III:}
\begin{enumerate}
     \item Compute the scaled $\hat{p}_{k,j}$ as in Rule I.
     \item Assign patients to the control arm with the same probability as the best arm. 
     \item Rescale all probabilities again and allocate the remaining samples to each arm with the corresponding rescaled probabilities. 
\end{enumerate}

\begin{remark}
Rule I ensures that the control arm is assigned with sufficient samples; Rule II gives more rewards to the best performing arm; Rule III well balances the best performing arm and the control arm to achieve the maximum power. Rescaling of $\hat{p}_{k,j}$ helps avoid extreme allocation probabilities. 
\end{remark}

\subsection{Arm suspension}\label{subsec:suspension}
When an arm has a very low posterior probability of being the best arm at some interim stage, we will consider suspending the arm, that is, no patients will be randomized to this arm at this stage. In our study, we consider using the threshold $\hat{p}_{k,j}<0.05$ as the criteria for suspension of treatment arm $k$ at the $j^{th}$ stage but other thresholds can also be assessed via sensitivity analysis. Note that the suspended arm may re-enter the study in a later stage. The control arm is never suspended regardless of its posterior probability. For the remaining arms, we re-allocate samples following the aforementioned RAR rules.

\subsection{Early stopping due to efficacy}\label{subsec:earlystop}

Study will stop early for efficacy if the posterior probability of difference between the best arm and control arm greater than a clinically meaningful difference is high, i.e., we evaluate the following at each interim stage $j$: 
\begin{equation}
\label{rule:eq2}
P(\theta_{best} - \theta_{0} >\delta | \mathfrak{D}_{j-1_+} ) > c_j ,
\end{equation} 
where $\theta_{best}= \max_{1\le k\le K} \theta_k$ is the mean transformed OSFD for the best arm, $\delta$ is the clinically meaningful difference chosen in our OSFD study. $c_j$ is the critical value used at $j^{th}$ interim stage. The critical value is changing across stages, we propose a procedure to calculate the corresponding critical value at each stage using different alpha spending functions to control for overall type I error. We use the posterior samples acquired from the Bayesian model on $\mathfrak{D}_{j-1_+}$ to estimate the corresponding posterior probabilities. 

Denote by $n_{j_+}$ the total number of patients assigned up to the $j^{th}$ stage and $t_j=n_{j_+}/N=j/J$ the information fraction at the $j^{th}$ interim analysis, in this study, we adapt the following three alpha spending functions, where $\alpha$ is the overall type I error rate. For a one-sided test, we set $\alpha=0.025$. 

\begin{enumerate}
  \item Pocock alpha spending function \citep{pocock1977group}: 
  $$ \alpha(t_j) = \alpha \log\{ 1 + (e-1)t_j  \}. $$ 
  \item O'Brien-Fleming (OF) alpha spending function \citep{o1979multiple}: 
  $$ \alpha(t_j) = 2 - 2\Phi(z_\frac{\alpha}{2} / \sqrt{t_j}), $$  where $\Phi$ is the CDF of the standard normal distribution.
  \item Power alpha spending function \citep{gordon1983discrete}: 
  $$ \alpha(t_j) = (t_j)^\gamma\alpha, $$ 
  where $\gamma$ is the power and in this study, we set $\gamma =1$.   
\end{enumerate}

At each interim stage $j$, we propose the simulation-based procedure summarized in Algorithm 1 to calculate the critical value $c_j$ with respect to the alpha spending function $\alpha(t_j)$. The procedure is adopted from the idea of \cite{zhu2017bayesian}. \\

\begin{algorithm}
\caption{Algorithm to calculate the critical value $c_j$}\label{alg1}
\begin{algorithmic}
\State \noindent\textbf{Data:} Simulate each of the $N_{rep}$ datasets $\mathfrak{D}^{m}=\{(D^{m}_{ij},\tau^{m}_{ij}):i=1,\ldots,n_{j}, j=1,\ldots,J \}$, $m=1,\ldots,N_{rep}$ for a treatment arm (any $k \in \{1,2,\ldots,K\}$) and the control arm ($k=0$) under $ H_0: \theta_{k}=\theta_{0}$, where $N_{rep}=10,000$, $n_j=200$, and $J=10$.\\
\State \noindent\textbf{Input:} $\mathbf{P}_1=(\vec{P}^1, \vec{P}^2,\ldots, \vec{P}^{N_{rep}})'$, $N_{rep} \times J$ matrix \\
\For $m=1, 2, \ldots, N_{rep}$ 
 \For $j=1,2,\ldots,J$
    \State  Calculate the posterior probability $P_j(\theta_{k} - \theta_{0} > \delta | \mathfrak{D}^{m}_{j-1_+} )$. 
    \State Set $\vec{P}^m= \{P_1(\theta_{k} - \theta_{0} > \delta), \ldots,P_J(\theta_{k} - \theta_{0} > \delta) \}' $ 
  \EndFor $j=1,2,\ldots,J$        
 \For $j = 1$
 \State Find $c_1= P^{\ast}(t_1)=\{1-\alpha(t_1)\}^{th} $ quantile of the first column of the matrix $\mathbf{P}_1$, where $t_1=1/J=1/10$. 
 \EndFor $j=1$
 \For $j = 2, 3, ..., J$
 \State Set $\mathbf{P}_j$ as a matrix composed of the rows of $\mathbf{P}_{j-1}$ such that the $(j– 1)^{th} $ element of the row be smaller than or equal to $P^{\ast}(t_{j-1})$.  
  \State Find $c_j = P^{\ast}(t_j) = \{1-\Delta \alpha (t_j)\}^{th}$ quantile of the $j^{th}$ column of matrix $\mathbf{P}_j$, where $\Delta \alpha (t_j) = \alpha (t_j) - \alpha (t_{j-1})$. 
 \EndFor $j = 2, 3, ..., J$    
  \EndFor $m=1, 2, \ldots, N_{rep}$
\State \noindent\textbf{Output:} $c_j$
\end{algorithmic}
\end{algorithm}

\subsection{Other RAR adaptation rules for comparison}

We compared our proposed RAR treatment allocation rules to the following three other rules that have been used in the Bayesian RAR design with different formulas to calculate the posterior probabilities $p_{k,j}$ at each interim stage $j$ under the same Bayesian framework:
\begin{itemize}
\item Fixed randomization (FR): 
$$p^{FR}_{k,j} = 1/(K+1) \text{ for } 0\le k\le K $$
\item Thompson sampling (TS) \citep{thompson1933likelihood}:
$$p^{TS}_{k,j} = \frac{P(\theta_{k} = \max_{0\le h\le K} \theta_{h,j}|\mathfrak{D})^c}{\sum\limits_{k=0}^K P(\theta_{k} = \max_{0\le h\le K} \theta_{h,j}|\mathfrak{D})^c } \text{ for } 0\le k\le K, $$
where $c=\frac{n_j}{2N}$. 
\item Trippa et al. procedure (TP) \citep{trippa2012bayesian}:
\begin{equation*}
p^{TP}_{k,j} = \frac{\pi_{k,j}}{\sum\limits_{k=0}^K \pi_{k,j}  }, \text{ where }
\pi_{k,j} =  
  \left\{ \begin{array}{ll} 
 \frac{P(\theta_{0}>\theta_{k}|\mathfrak{D})^{\gamma_j}}{\sum\limits_{k=1}^K P(\theta_{0}>\theta_{k}|\mathfrak{D})^{\gamma_j} }, \hspace{0.1cm} k=1,\ldots, K \\ 
 \frac{1}{K} \exp[\max_{1\le k \le K} n_{k,j-1_+} - n_{0,j-1_+} ]^{\eta_j}, \hspace{0.1cm} k=0,  \\
\end{array} 
\right.
\end{equation*}

where $\gamma_j= 10(n_j/N)^{0.75}$, $\eta_j = 0.25(n_j/N)$ and $n_{k,j-1_+}$ is the number of samples in the $k$th arm collected up to stage $j-1$. 
\end{itemize}

\section{Final analysis and assessment}\label{sec:final}
In the final analysis, we assessed the proportion of patients assigned to the best arm(s), i.e., patient benefits, and the posterior probability of superiority (PPS) of the trial. PPS can be served as an alternative to statistical power, the higher PPS the more powered a method is. With one treatment arm $k$, this corresponds to the difference in mean transformed OSFD between treatment arm and control arm greater than a clinically meaningful threshold $\delta$: 
\begin{equation}
\label{final:eq1}
PPS = P(\theta_{k} - \theta_{0} >\delta | \mathfrak{D}_J ) ,
\end{equation} 
where $\mathfrak{D}_J$ is all the data collected up to the end of the study. The PPS is estimated over $R$ Monte Carlo replications as: 
\begin{equation*}
\widehat{PPS}=\frac{1}{R} \sum\limits_{r=1}^R \hat{P}_r(\theta_{k} - \theta_{0} >\delta | \mathfrak{D}_J ), 
\end{equation*}
where $R$ is the number of replications and $\hat{P}_r(\theta_{k} - \theta_{0} >\delta | \mathfrak{D}_J ) = \sum\limits_{b=1}^B \sum\limits_{b'=1}^B I(\theta^{(rb)}_k - \theta^{(rb')}_0 > \delta) / B^2$ is the posterior estimate of success probability in $r$th replication. 

In our study, we consider a total of $K$ treatment arms and one control arm. Depending on the type of the alternative hypothesis ($H_{a1}$, $H_{a2}$ or $H_{a3}$), the posterior probability of superiority can be calculated accordingly: 
$$ PPS(H_{a1}) = P(\theta_{best}- \theta_{0} > \delta | \mathfrak{D}_J); $$ 
$$ PPS(H_{a2}) = P\big\{\sum\limits_{k=1}^K I( \theta_{k} - \theta_{0} >\delta)\ge r | \mathfrak{D}_J\big\}; $$ 
$$ PPS(H_{a3}) = P(\theta_{worst}- \theta_{0} > \delta | \mathfrak{D}_J). $$ 
 
\section{Simulation studies}\label{sec:sim}

\subsection{Setting}\label{subsec:setting}
In this section, we evaluate the performance of our proposed Bayesian model and adaptation rules via simulations. We assumed a total sample size of $N=2,000$ patients in the trial. For each of the $J=10$ interim analyses, we will perform the adaptation of the treatment allocation rule. In the simulation, we assumed a total of $K=4$ treatment comparison arms with three active treatments and one control. For the control arm, we assumed $\mu_{0} =-2.3$ and $\sigma_{0} =0.8$ for the non-censored patients and a mortality rate of $\lambda_{0}= 20\%$ based on a pilot trial \citep{annane2018hydrocortisone}. Comparing to the control arm, we considered treatment arms with strong (S) effect size (Cohen's d=0.8) and medium (M) effect size (Cohen's d=0.5) and no effect (N). A strong effect corresponds to mean OSFD difference of 3.5 on the original scale, while a medium effect corresponds to mean OSFD difference of 2.5. We also considered different mortality rates in these arms, where $\lambda_k = 15\%$ for arms with strong effect and $\lambda_k = 18\%$ for medium and $\lambda_k = 20\%$ for no effect. The standard deviations for all treatment arms were kept the same as the control arm: $\sigma_{k}=0.8, 1\le k\le K$. We considered two levels of censoring proportions $\omega_{k} \in \{0.2,0.3\}, 0\le k\le K$ for comprehensive evaluation. Here we kept the censoring proportion the same for all arms. In real data, the censoring proportions may vary across arms and our model is flexible to handle such variation. We simulated the following five scenarios with varying treatment effects as follows: 
\begin{enumerate}
\item ``MNN'': one arm has medium effect, but all other arms have no effect.   
\item ``SNN'': one arm has strong effect, but all other arms have no effect.   
\item ``SMN'': one arm has strong effect and one arm has medium effect, and the remaining third arm has no effect.  
\item ``SMM'': one arm has strong effect, and the other two arms have medium effects.  
\item ``SSM'': two arms have strong effects, and the remaining arm has medium effect. 
 \end{enumerate}
In our simulation, we set the clinically difference $\delta=0.67$ and $\delta=0.65$ for $\omega_k=0.2$ and $\omega_k=0.3$, respectively. This corresponds to a mean OSFD difference of 1.5 on the original scale. 
 
We evaluated the posterior probability of superiority of the trial as well as the proportion of patients assigned to the best arm with our proposed adaptation rules. The posterior probability of superiority was assessed for all three alternative hypotheses: $H_{a1}$: at least one arm has effect; $H_{a2}$: at least two arms have effect; $H_{a3}$: all three arms have effect. All results were based on 5,000 replications of the trials. For each trial, we ran 2,000 MCMC samples and removed the first 500 samples in the burn-in period, and took every $10^{th}$ sample as the posterior samples for our final inference.

\subsection{Results}\label{subsec:results}
Figure 3 shows the proportion of patients assigned to the best performing arm(s) (with standard error) under different allocation rules in various simulation scenarios when $\omega=0.3$. When there were two arms with strong effects, we considered both arms as best performers and calculated the sum of proportion of patients assigned to both arms. 
Overall, SSM, SNN and MNN have a higher proportion of samples assigned to best arm(s). The proportion was higher when using our proposed adaptive allocation rules as compared to the competing rules in all scenarios. Rule II has the highest proportion among all allocation rules, around 10\% higher than Rule I and 25\% higher than Rule III on average. Rule II assigns patients to the best arm first at interim stages to ensure it has sufficient samples thus has the highest proportion. Under SSM with two best arms, the advantage of our proposed adaptive rules over other allocation rules was becoming smaller as TP and fixed randomization started to assign more patients to either of two arm(s). For alpha spending functions, using the OBF method resulted in more patients being assigned to the best arm(s) as compared to Pocock and power methods. 

\begin{figure}
\centering
\includegraphics[scale=0.55]{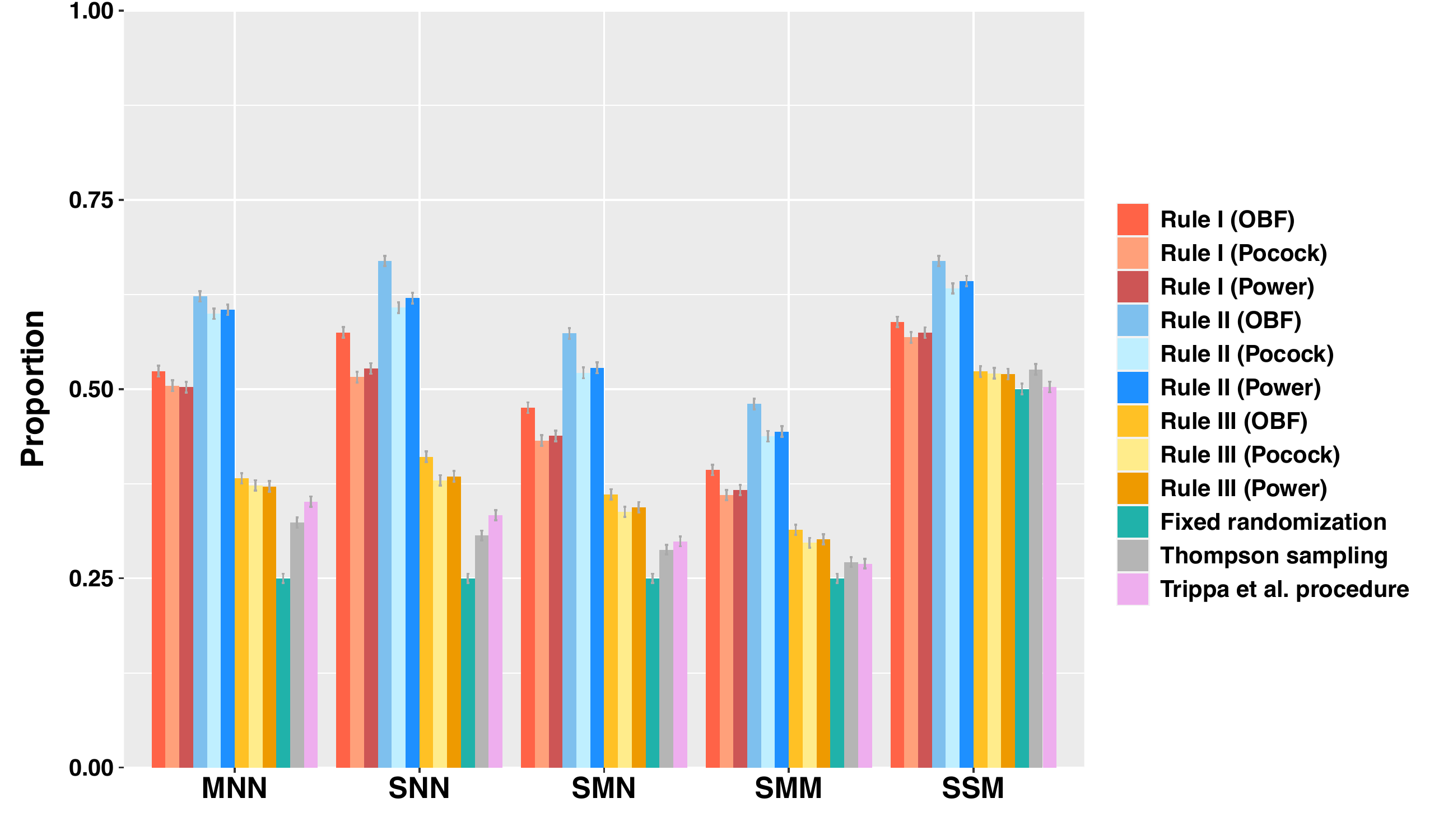}
\caption{The proportion of total patients assigned to the best performing arm(s) with standard error under different allocation rules when $\omega$ = 0.3.}
\label{fig3}
\end{figure}

Table 1 compares the posterior probability of superiority for all scenarios under different hypotheses settings when $\omega=0.3$. When there is only one dominating treatment arm (e.g. MNN and SNN), we only evaluated $H_{a1}$. When there are more than one dominating treatment arm (e.g. SMN, SMM and SSM), we also evaluated $H_{a2}$ and $H_{a3}$. In general, scenarios with strong effect arms and more contrast between treatment arms (e.g. SNN) have overall greater posterior probability of superiority. In all scenarios, our proposed Rule I and Rule III are more powered than the other allocation rules for the same alternative hypothesis. Rule III achieved the greatest posterior probability of superiority in most cases with a more balanced sample assignment between the best and the control arms. Though with the highest proportion assigned to the best arm(s), Rule II had the lowest posterior probability of superiority among the three adaptive allocation rules we proposed, but still a similar posterior probability of superiority as the other rules, i.e., Thompson sampling (TS) and Trippa et al. procedure (TP). 

\begin{table}[]
\footnotesize
\centering
\caption{Comparison of the posterior probability of superiority under different hypotheses settings.}
\begin{tabular}{p{0.4cm}p{0.3cm} p{0.7cm}p{0.7cm}p{0.7cm} p{0.7cm}p{0.7cm}p{0.7cm}  p{0.7cm}p{0.7cm}p{0.7cm}   p{0.7cm}p{0.7cm}p{0.7cm} }
\hline
    &     & \multicolumn{3}{l}{Rule I}     & \multicolumn{3}{l}{Rule II}    & \multicolumn{3}{l}{Rule III}   & \multirow{2}{*}{\begin{tabular}[c]{@{}l@{}} \\ FR \end{tabular}} & \multirow{2}{*}{\begin{tabular}[c]{@{}l@{}}\\TS \end{tabular}} & \multirow{2}{*}{\begin{tabular}[c]{@{}l@{}}\\ TP\end{tabular}} \\ \cline{3-11}
    &     & OBF    & Pocock & Power & OBF    & Pocock & Power & OBF    & Pocock & Power &                                                                                 &                                                                                    &                                                                                          \\ \hline
MNN & $H_{a1}$ & 0.805   & 0.810   & 0.811        & 0.756   & 0.760   & 0.763        & 0.827   & 0.834   & 0.825        & 0.797                                                                           & 0.809                                                                              & 0.823                                                                                    \\
    &     & (0.228) & (0.226) & (0.225)      & (0.249) & (0.251) & (0.248)      & (0.217) & (0.213) & (0.222)      & (0.225)                                                                         & (0.213)                                                                            & (0.221)                                                                                  \\ \hline
SNN & $H_{a1}$ & 0.964   & 0.963   & 0.961        & 0.937   & 0.936   & 0.935        & 0.974   & 0.970   & 0.970        & 0.929                                                                           & 0.940                                                                              & 0.950                                                                                    \\
    &     & (0.071) & (0.071) & (0.076)      & (0.111) & (0.115) & (0.113)      & (0.054) & (0.066) & (0.065)      & (0.100)                                                                         & (0.077)                                                                            & (0.072)                                                                                  \\ \hline
SMN & $H_{a1}$ & 0.955   & 0.957   & 0.956        & 0.921   & 0.925   & 0.924        & 0.966   & 0.963   & 0.965        & 0.937                                                                           & 0.936                                                                              & 0.946                                                                                    \\
    &     & (0.088) & (0.084) & (0.083)      & (0.131) & (0.128) & (0.126)      & (0.075) & (0.076) & (0.069)      & (0.083)                                                                         & (0.085)                                                                            & (0.078)                                                                                  \\
    & $H_{a2}$ & 0.671   & 0.690   & 0.687        & 0.647   & 0.659   & 0.659        & 0.687   & 0.701   & 0.694        & 0.691                                                                           & 0.686                                                                              & 0.669                                                                                    \\
    &     & (0.257) & (0.247) & (0.248)      & (0.268) & (0.265) & (0.265)      & (0.251) & (0.242) & (0.244)      & (0.225)                                                                         & (0.230)                                                                            & (0.243)                                                                                  \\ \hline
SMM & $H_{a1}$ & 0.925   & 0.926   & 0.925        & 0.899   & 0.899   & 0.896        & 0.933   & 0.930   & 0.931        & 0.908                                                                           & 0.900                                                                              & 0.906                                                                                    \\
    &     & (0.146) & (0.145) & (0.141)      & (0.160) & (0.160) & (0.165)      & (0.143) & (0.146) & (0.141)      & (0.139)                                                                         & (0.150)                                                                            & (0.155)                                                                                  \\
    & $H_{a2}$ & 0.737   & 0.751   & 0.743        & 0.715   & 0.725   & 0.719        & 0.754   & 0.765   & 0.756        & 0.743                                                                           & 0.732                                                                              & 0.721                                                                                    \\
    &     & (0.226) & (0.214) & (0.219)      & (0.238) & (0.233) & (0.239)      & (0.218) & (0.209) & (0.211)      & (0.193)                                                                         & (0.203)                                                                            & (0.217)                                                                                  \\
    & $H_{a3}$ & 0.471   & 0.507   & 0.492        & 0.473   & 0.495   & 0.493        & 0.474   & 0.512   & 0.498        & 0.524                                                                           & 0.506                                                                              & 0.474                                                                                    \\
    &     & (0.278) & (0.275) & (0.277)      & (0.280) & (0.281) & (0.284)      & (0.277) & (0.277) & (0.277)      & (0.243)                                                                         & (0.252)                                                                            & (0.266)                                                                                  \\ \hline
SSM & $H_{a1}$ & 0.947   & 0.946   & 0.948        & 0.922   & 0.921   & 0.922        & 0.957   & 0.952   & 0.951        & 0.934                                                                           & 0.931                                                                              & 0.934                                                                                    \\
    &     & (0.106) & (0.103) & (0.103)      & (0.130) & (0.134) & (0.130)      & (0.088) & (0.097) & (0.102)      & (0.093)                                                                         & (0.094)                                                                            & (0.104)                                                                                  \\
    & $H_{a2}$ & 0.821   & 0.822   & 0.825        & 0.801   & 0.799   & 0.797        & 0.841   & 0.832   & 0.834        & 0.813                                                                           & 0.806                                                                              & 0.800                                                                                    \\
    &     & (0.195) & (0.186) & (0.186)      & (0.211) & (0.209) & (0.207)      & (0.184) & (0.180) & (0.182)      & (0.168)                                                                         & (0.173)                                                                            & (0.190)                                                                                  \\
    & $H_{a3}$ & 0.565   & 0.592   & 0.587        & 0.572   & 0.583   & 0.578        & 0.578   & 0.593   & 0.590        & 0.600                                                                           & 0.587                                                                              & 0.564                                                                                    \\
    &     & (0.271) & (0.265) & (0.265)      & (0.273) & (0.273) & (0.273)      & (0.273) & (0.265) & (0.267)      & (0.240)                                                                         & (0.248)                                                                            & (0.261)                                                                                   \\ \hline
\end{tabular}
\flushleft Note: $H_{a1}$: at least one arm has treatment effect; $H_{a2}$: at least two arms have treatment effects; $H_{a3}$: all three arms have treatment effects. \\
MNN: one arm has medium effect, but all other arms have no effect.   \\
SNN: one arm has strong effect, but all other arms have no effect.   \\
SMN: one arm has strong effect and one arm has medium effect, and the remaining third arm has no effect.  \\
SMM: one arm has strong effect, and the other two arms have medium effects.  \\
SSM: two arms have strong effects, and the remaining arm has medium effect. \\
OBF: O'Brien-Fleming alpha spending function. \\
FR: Fixed randomization; TS: Thompson sampling; TP: Trippa et al. procedure.
 
\end{table}

\section{Discussion and Conclusion}

In this paper, we proposed a Bayesian response adaptive randomization design for a composite endpoint combining both mortality and morbidity. We proposed a two-component mixture model for this mixed types of outcome under a Bayesian framework and used Gibbs sampling to estimate the posterior distribution of the parameters. At each interim stage, the treatment allocation probability for each arm can be estimated based on the posterior probability that the arm is the best arm given the data collected up to the current stage. We conducted extensive simulations with varying treatment effects in different arms to evaluate the adaptive allocation rules we proposed. The results showed that our proposed allocation rules are more powerful and can allocate more patients to the best arm than the other existing allocation rules. We modified the Bayesian alpha spending functions proposed by \cite{zhu2017bayesian} to our design and applied it to control the overall false positive error rate. Simulation results also showed that more patients are assigned to the superior arm using our proposed adaptive allocation rules, with a greater posterior probability of superiority obtained at the same time. 

Our Bayesian approach assumes a mixture of spike and truncated normal to model the observed endpoint based on the exploration of pilot data from the Sepsis ENdotyping in Emergency Care (SENECA) project. In any chance the distributional assumption is violated, we can modify the model using nonparametric Bayesian Dirichlet mixture model or apply other popular nonparametric allocation rules such as win ratio \citep{pocock2012win}. 

Our methods do not take patient's baseline characteristics (e.g., disease stage, race) into consideration. To further investigate the impact of patient characteristics on treatment allocation rule, we will incorporate these baseline covariates in the Bayesian model in the future development. 

We developed an R package, BRACE, which is publicly available at https://github.com/joyxuuu/BRACE for implementation of our method.


\section*{Acknowledgments}
This is acknowledgment text. Provide text here. This is acknowledgment text. Provide text here. This is acknowledgment text. Provide text here. This is acknowledgment text. Provide text here. This is acknowledgment text. Provide text here. This is acknowledgment text. Provide text here. This is acknowledgment text. Provide text here. This is acknowledgment text. Provide text here. This is acknowledgment text. Provide text here. 

\subsection*{Author contributions}

This is an author contribution text. This is an author contribution text. This is an author contribution text. This is an author contribution text. This is an author contribution text. 

\subsection*{Financial disclosure}

None reported.

\subsection*{Conflict of interest}

The authors declare no potential conflict of interests.

\section*{Supporting information}

The following supporting information is available as part of the online article:

\noindent
\textbf{Figure S1.}
{500{\uns}hPa geopotential anomalies for GC2C calculated against the ERA Interim reanalysis. The period is 1989--2008.}

\noindent
\textbf{Figure S2.}
{The SST anomalies for GC2C calculated against the observations (OIsst).}

\appendix

\section{. Full conditional posterior distribution for Gibbs sampling}

We used Gibbs sampling to update all parameters in the model, the corresponding full conditional posterior distributions for each parameter are summarized below.

$$ (\lambda_{k} | \mu_{k}, \sigma^2_{k},\omega_{k}, z_{ik}, D_{ik}, \tau_{ik}) \propto  Beta\big( 1+ \sum_i \tau_{ik}, 1+ n - \sum_i \tau_{ik} \big) $$

$$ (\omega_{k} | \lambda_{k}, \mu_{k}, \sigma^2_{k}, z_{ik}, D_{ik},\tau_{ik}) \propto Beta\big( 1+ \sum_i z_{ik}, 1+ n - \sum_i z_{ik} \big)  $$

$$ (z_{ik} | \lambda_{k}, \mu_{k}, \sigma^2_{k}, \omega_{(k)}, D_{ik},\tau_{ik}) \propto Ber \big[ \frac{\omega_{k}}{\omega_{k}+(1-\omega_{k})\times TN\{D_{ik} |\mu_{k}, \sigma^2_{k},\lceil \log(30) \rceil-\log(30),\lceil \log(30) \rceil-\log(2)\}} \big]  $$

$$ (\mu_{k} | \lambda_{k} , \sigma^2_{k}, \omega_{k}, z_{ik}, D_{ik},\tau_{ik}) \propto  N \big\{  \frac{1}{\frac{1}{10^4} + \frac{n}{\sigma^2_{k}}} ( \frac{{\sum_{i=1}^n}D_{ik} } {\sigma^2_{k}}) , \frac{1}{{ \frac{1}{10^4 }} + \frac{n}{\sigma^2_{k}} } \big\} $$   

$$ (\sigma^2_{k}| \lambda_{k} , \mu_{k} , \omega_{k}, z_{ik}, D_{ik},\tau_{ik}) \propto IG \big\{  10^{-4} + \frac{n}{2}, 10^{-4} + \frac{ {\sum_{i=1}^n}{( D_{ik} - \mu_{k})^2} }{2}  \big\}.$$

\section{. Derivation of variance of transformed OSFD}

The variance of $D_{ik}$ can be derived as follows:

\begin{equation}
\begin{split} 
Var(D_{ik}) & = E\big[\{D_{ik}-E(D_{ik})\}^2|\tau_{ik}=0)\big]P(\tau_{ik}=0)  + E\big[\{D_{ik}-E(D_{ik})\}^2|\tau_{ik}=1)\big]P(\tau_{ik}=1)  \\ 
          &  = \big[E(D^2_{ik}|\tau_{ik}=0)- E^2(D_{ik}|\tau_{ik}=0)\big](1-\lambda_k) + 0 * \lambda_k \\  
          &  = \big [ \omega_k*\{\lceil \log(30) \rceil-\log(2)\}^2 + (1 - \omega_k )(E^2_{TN(\mu_k,\sigma^2_k)} + Var_{TN(\mu_k,\sigma^2_k)})  - \theta^2_k\big] (1-\lambda_k)  \\ 
          &  = \big[\{\lceil \log(30) \rceil-\log(2)\}^2\omega_k + (1 - \omega_k )\{ (\mu_k - \sigma_k  R_{TN,k})^2 + \sigma^2_k (1+ {WR}_{TN,k} - R^2_{TN,k})  \}     \\ 
          & \hspace{0.3cm}  - \{ \lceil \log(30) \rceil-\log(2)\omega_k + (1-\omega_k)( \mu_k - \sigma_k  R_{TN,k} ) \}^2 \big] (1-\lambda_k) , \\
 \end{split}
\end{equation}
where 
$R_{TN,k}= \frac{  \phi\{\frac{\lceil \log(30) \rceil-\log(2)-\mu_k}{\sigma_k}\}  - \phi\{\frac{\lceil \log(30) \rceil-\log(30)-\mu_k}{\sigma_k}\} } { \Phi\{\frac{\lceil \log(30) \rceil-\log(2)-\mu_k}{\sigma_k}\}  - \Phi\{\frac{\lceil \log(30) \rceil-\log(30)-\mu_k}{\sigma_k}\}  } $ and \\
${WR}_{TN,k} = \frac{ \{\frac{\lceil \log(30) \rceil-\log(2)-\mu_k}{\sigma_k}\} \phi\{\frac{\lceil \log(30) \rceil-\log(2)-\mu_k}{\sigma_k}\}  - \{\frac{\lceil \log(30) \rceil-\log(30)-\mu_k}{\sigma_k}\} \phi\{\frac{\lceil \log(30) \rceil-\log(30)-\mu_k}{\sigma_k}\} } { \Phi\{\frac{\lceil \log(30) \rceil-\log(2)-\mu_k}{\sigma_k}\}  - \Phi\{\frac{\lceil \log(30) \rceil-\log(30)-\mu_k}{\sigma_k}\}  } $.

\section{. Proof of theorems}

\begin{lemma}{Doob's consistency Theorem \citep{doob1953stochastic}}. Suppose $X_1,\ldots,X_n$ are i.i.d. drawn from a correctly specified model $\{p(X|\gamma),\gamma \in \Upsilon \}$ with $\mathbf{P}_{\gamma^*}$ for some true parameter $\gamma^*$ and $\gamma^*$ is drawn from the prior $\mathbf{\pi}(\gamma)$, then the posterior mean $\hat{\gamma} \overset{p}{\to} \gamma^*$.  
\end{lemma}

\begin{remark}
A proof of this theorem is provided in \cite{van2000asymptotic}, we do not repeat it here. Doob's consistency Theorem implies that the posterior mean of a parameter with finite dimension in a correctly specified model is consistent.
\end{remark}

\begin{lemma} {Bernstein-von Mises Theorem \citep{van2000asymptotic}}. For a well-specified regular parametric model $\{ p(X|\Delta), \gamma \in \Upsilon \}$ with $\mathbf{P}_{\gamma^*}$ under certain regularity conditions and with nonsingular Fisher information matrix $\mathbf{I}_{\gamma^*}$, suppose the prior density $\pi(\gamma)$ is continuous in a neighborhood of $\gamma^*$, then the corresponding posterior distribution satisfy:
  
$$||\mathbf{P}_{\sqrt{n}(\hat{\gamma}_n - \gamma^* )|X_1,\ldots,X_n} - N(\nu_{n,\gamma^*},\mathbf{I}^{-1}_{\gamma^*} ) || \overset{\mathbf{P}^n_{\gamma^*}}{\to} 0. $$  

\noindent As $n \rightarrow \infty$, $\nu_{n,\gamma^*}$ weakly converges to $N(0,\mathbf{I}^{-1}_{\gamma^*} )$. 
\end{lemma}

\begin{remark}
A proof of this theorem is provided in \cite{van2000asymptotic}. We do not repeat it here. The Bernstein-von Mises theorem links Bayesian inference with frequentist inference. It shows that the posterior laws converge in distribution to a Gaussian posterior law in total variation distance. This immediately implies that any location functional suitably continuous relative to the total variation norm applied to the sequence of posterior law converges to the same location functional applied to the limiting Gaussian posterior distribution. For most choices, this means $N(0,\mathbf{I}^{-1}_{\gamma^*} )$. Suppose the posterior mean is defined as $\hat{\gamma}= \int \gamma p(\gamma|X) d\gamma$. An immediate result from the Bernstein-von Mises theorem implies the asymptotic normal distribution of $\hat{\gamma}$, i.e., $\sqrt{n} (\hat{\gamma} - \gamma^*) \sim N(0, \mathbf{I}^{-1}_{\gamma^*}) $. Basically, the Bernstein-von Mises theorem establishes that $\hat{\gamma}$ is on the same footing as the maximum likelihood estimate of $\gamma$.   
\end{remark}

\subsection{Proof of Theorem 2}
\begin{proof}
The model meets the general regularity conditions and the prior $\pi(\Delta)$ is continuous, so the asymptotic normality results of $\hat{\Delta}_k$ directly follow the Bernstein-von Mises theorem, i.e., $\sqrt{n} (\hat{\Delta}_k - \Delta^*) \overset{d}{\to} N(0,\mathbf{I}^{-1}_{\Delta^*})$. 
What we show here is the derivation of the nonsingular Fisher's information matrix thus the asymptotic variance.

We first write out the full likelihood of the model (let $\Delta_k = (\lambda_k,\omega_k,\mu_k,\sigma^2_k)$):
\begin{equation}
 \begin{split}
 \mathcal{L}(\Delta_k|D_{ik},\tau_{ik}) & = \prod\limits_{i=1}^n f(D_{ik},\tau_{ik})  \\
   & =   \prod\limits_{i=1}^n \{ f(D_{ik}|\tau_{ik}=0)P(\tau_{ik}=0) + f(D_{ik}|\tau_{ik}=1)P(\tau_{ik}=1) \} \\
   & =  \prod\limits_{i=1}^n (1-\lambda_k)\big[ \omega_k + (1-\omega_k) TN\{D_{ik};\mu_k,\sigma^2_k,\lceil \log(30) \rceil-\log(30),\lceil \log(30) \rceil-\log(2)\}\big].
 \end{split}
\end{equation}

Taking the log, we get the log-likelihood:
\begin{equation}
\ell(\Delta_k|D_{ik},\tau_{ik})=n\log (1-\lambda_k) + \sum\limits_{i=1}^n \log \big[ \omega_k + (1-\omega_k) TN\{D_{ik};\mu_k,\sigma^2_k,\lceil \log(30) \rceil-\log(30),\lceil \log(30) \rceil-\log(2)\}\big], 
\end{equation}
where $TN\{D_{ik};\mu_k,\sigma^2_k,\lceil \log(30) \rceil-\log(30),\lceil \log(30) \rceil-\log(2)\} = \frac{\phi(\frac{D_{ik}-\mu_{k}}{\sigma_k})}{\sigma_k\big[\Phi\{\frac{\lceil \log(30) \rceil-\log(2)-\mu_k}{\sigma_k}\}- \Phi\{\frac{\lceil \log(30) \rceil-\log(30)-\mu_k}{\sigma_k}\}    \big]}$, $\phi(.)$ and $\Phi(.)$ stand for the PDF and CDF of the standard normal distribution, respectively. 

We then derive the score function:
\begin{equation}
 \begin{split}
 S(\Delta_k) & = \begin{pmatrix} 
                      \frac{\partial \ell(\Delta_k) }{\partial \lambda_k } & \\
                      \frac{\partial \ell(\Delta_k) }{\partial \omega_k } & \\
                      \frac{\partial \ell(\Delta_k) }{\partial \mu_k } & \\
                      \frac{\partial \ell(\Delta_k) }{\partial \sigma^2_k } & \\
                  \end{pmatrix}  =  \begin{pmatrix} 
                       - \frac{ n }{ \lambda_k } & \\
   \sum\limits_{i=1}^n   \frac{ 1- \psi({D_{ik}}) }{  \omega_k + (1-\omega_k)\psi({D_{ik}}) } & \\
   \sum\limits_{i=1}^n   \frac{ (1- \omega_k)\psi'_1({D_{ik}}) }{  \omega_k + (1-\omega_k)\psi({D_{ik}}) } & \\
      \sum\limits_{i=1}^n   \frac{ (1- \omega_k)\psi'_2({D_{ik}}) }{  \omega_k + (1-\omega_k)\psi({D_{ik}}) } & \\
                   \end{pmatrix}  \\    
 \end{split},
\end{equation}
where $\psi(.)$ is the PDF of $TN(\mu_k,\sigma^2_k,\lceil \log(30) \rceil-\log(30),\lceil \log(30) \rceil-\log(2))$, $\psi'_1({D_{ik}}) = \frac{\partial \psi({D_{ik}}) }{\partial \mu_k }$, $\psi'_2({D_{ik}}) = \frac{\partial \psi({D_{ik}}) }{\partial \sigma^2_k }$. We also know for each ith case, $s_i(\Delta_k)= \{ - \frac{ 1 }{ \lambda_k }, \frac{ 1- \psi({D_{ik}}) }{  \omega_k + (1-\omega_k)\psi({D_{ik}}) }, \frac{ (1- \omega_k)\psi'_1({D_{ik}}) }{  \omega_k + (1-\omega_k)\psi({D_{ik}}) }, \frac{ (1- \omega_k)\psi'_2({D_{ik}}) }{  \omega_k + (1-\omega_k)\psi({D_{ik}}) }  \}^T $. 

Assuming i.i.d., we can further derive Fisher's information matrix:
\begin{equation}
 \begin{split}
& I(\Delta_k) = - E\left\{\frac{\partial s_i(\Delta_k)}{\partial \Delta_k }\right\} = E \begin{pmatrix} 
              - \frac{1}{\lambda^2_k} &  0 &  0 &  0 \\ 
            0 &  I_{i,22} &  I_{i,23} &  I_{i,24}   \\
0 &  I_{i,23}  & I_{i,33}  &  I_{i,34} \\                                                
0 &  I_{i,24} &  I_{i,34} & I_{i,44} \\ 
                  \end{pmatrix}      
 \end{split},
\end{equation}
where $\psi_i = \psi({D_{ik}})$, $\psi'_{i1} = \psi'_1({D_{ik}})$, $\psi'_{i2} = \psi'_2({D_{ik}})$, $\psi_{i1}^{''}= \frac{\partial \psi'_1({D_{ik}})}{ \partial \mu_k}$,$\psi_{i2}^{''}= \frac{\partial \psi'_2({D_{ik}})}{ \partial \sigma^2_k}$, $\psi_{i12}^{''}= \frac{\partial \psi({D_{ik}})}{ \partial \mu_k \partial \sigma^2_k}$, $W_{\psi}=\omega_k +(1-\omega_k)\psi_i$, $I_{i,22}=\frac{(1-\psi_i)^2}{W_\psi^2}$, $I_{i,23}=\frac{W_\psi(1-\psi'_{i1})-(1-\omega_k)(1-\psi_i) \psi'_{i1} }{W_\psi^2}$, $I_{i,24}=\frac{W_\psi(1-\psi'_{i2})-(1-\omega_k)(1-\psi_i) \psi'_{i2} }{W_\psi^2}$, $I_{i,33}=\frac{W_\psi(1-\omega_k)\psi^{''}_{i1}-(1-\omega_k)^2\psi^{''2}_{i1}}{W_\psi^2}$, $I_{i,34}=\frac{W_\psi(1-\omega_k)\psi^{''}_{i2}-(1-\omega_k)^2\psi'_{i1} \psi'_{i2}}{W_\psi^2}$, and $I_{i,44}=\frac{W_\psi(1-\omega_k)\psi^{''}_{i2}-(1-\omega_k)^2\psi^{''2}_{i2}}{W_\psi^2}$. We have completed the derivation of a nonsingular information matrix $I(\Delta_k)$. 

Since $\hat{\theta}_k=g(\hat{\Delta}_k)$ is a continuous function of $\hat{\Delta}_k$, we can apply the multivariate Delta's method and obtain the asymptotic normal distribution of $\hat{\theta}_k$: $ \sqrt{n} (\hat{\theta}_k - \theta^*_k)  \overset{d}{\to} N\big\{0, \nabla \theta_k(\Delta_k)^T \mathbf{I}^{-1}_{\Delta^*}  \nabla \theta_k(\Delta_k) \big\} $. 

What we need to do here is just to derive $\nabla \theta_k(\Delta_k)$:
\begin{equation}
 \begin{split}
  \nabla \theta_k(\Delta_k) & = \begin{pmatrix} 
                      \frac{\partial \theta_k }{\partial \lambda_k } & \\
                      \frac{\partial \theta_k }{\partial \omega_k } & \\
                      \frac{\partial \theta_k }{\partial \mu_k } & \\
                      \frac{\partial \theta_k }{\partial \sigma^2_k } & \\
                  \end{pmatrix}
                = \begin{pmatrix} 
                     -\big[\{\lceil \log(30) \rceil-\log(2)\}\omega_k + (1-\omega_k)( \mu_k - \sigma_k R_{TN})\big]  & \\
                      (1-\lambda_k)\big\{ \lceil \log(30) \rceil-\log(2) - (\mu_k - \sigma_k R_{TN})\big\}  & \\
                     (1- \lambda_k)(1-\omega_k) ( 1 - \sigma_k R^{1'}_{TN} )  & \\                   
                    (1- \lambda_k)(1-\omega_k) ( - \frac{1}{2 \sigma_k} R_{TN} -  \sigma_k R^{2'}_{TN} )  & \\
                  \end{pmatrix}     
 \end{split},
\end{equation}
where $R_{TN,k}= \frac{  \phi\{\frac{\lceil \log(30) \rceil-\log(2)-\mu_k}{\sigma_k}\}  - \phi\{\frac{\lceil \log(30) \rceil-\log(30)-\mu_k}{\sigma_k}\} } { \Phi\{\frac{\lceil \log(30) \rceil-\log(2)-\mu_k}{\sigma_k}\}  - \Phi\{\frac{\lceil \log(30) \rceil-\log(30)-\mu_k}{\sigma_k}\}  } $, $R^{1'}_{TN,k} = \frac{\partial R_{TN,k}}{\partial \mu_k}$, and $R^{2'}_{TN,k} = \frac{\partial R_{TN,k}}{\partial \sigma^2_k}$.

\end{proof}

\bibliography{BayesianRAR_SIM}%

\begin{thebibliography}{10}
\providecommand \doibase [0]{http://dx.doi.org/}%

\bibitem{lin2016general}
Lin J, Lin LA, Sankoh S. A general overview of adaptive randomization design
  for clinical trials. {\it J Biom Biostat} 2016\string; 7(2)\string: 294.

\bibitem{rosenberger2001optimal}
Rosenberger WF, Stallard N, Ivanova A, Harper CN, Ricks ML. Optimal adaptive
  designs for binary response trials. {\it Biometrics} 2001\string;
  57(3)\string: 909--913.

\bibitem{wei1978randomized}
Wei L, Durham S. The randomized play-the-winner rule in medical trials. {\it
  Journal of the American Statistical Association} 1978\string; 73(364)\string:
  840--843.

\bibitem{ivanova2003play}
Ivanova A. A play-the-winner-type urn design with reduced variability. {\it
  Metrika} 2003\string; 58(1)\string: 1--13.

\bibitem{villar2015multi}
Villar SS, Bowden J, Wason J. Multi-armed bandit models for the optimal design
  of clinical trials: benefits and challenges. {\it Statistical science: a
  review journal of the Institute of Mathematical Statistics} 2015\string;
  30(2)\string: 199.

\bibitem{hu2004asymptotic}
Hu F, Zhang LX. Asymptotic properties of doubly adaptive biased coin designs
  for multitreatment clinical trials. {\it The Annals of Statistics}
  2004\string; 32(1)\string: 268--301.

\bibitem{eisele1995central}
Eisele JR, Woodroofe MB. Central limit theorems for doubly adaptive biased coin
  designs. {\it The Annals of Statistics} 1995\string: 234--254.

\bibitem{angus2020optimizing}
Angus DC. Optimizing the trade-off between learning and doing in a pandemic.
  {\it Jama} 2020\string; 323(19)\string: 1895--1896.

\bibitem{magaret2020multigroup}
Magaret AS, Jacob ST, Halloran ME, et al. Multigroup, adaptively randomized
  trials are advantageous for comparing coronavirus disease 2019 (COVID-19)
  interventions. {\it Annals of Internal Medicine} 2020\string; 173(7)\string:
  576--577.

\bibitem{villar2021temptation}
Villar SS, Robertson DS, Rosenberger WF. The temptation of overgeneralizing
  response-adaptive randomization. {\it Clinical Infectious Diseases}
  2021\string; 73(3)\string: e842--e842.

\bibitem{anand2009cardiac}
Anand IS, Carson P, Galle E, et al. Cardiac resynchronization therapy reduces
  the risk of hospitalizations in patients with advanced heart failure: results
  from the Comparison of Medical Therapy, Pacing and Defibrillation in Heart
  Failure (COMPANION) trial. {\it Circulation} 2009\string; 119(7)\string:
  969--977.

\bibitem{cannon2004intensive}
Cannon CP, Braunwald E, McCabe CH, et al. Intensive versus moderate lipid
  lowering with statins after acute coronary syndromes. {\it New England
  journal of medicine} 2004\string; 350(15)\string: 1495--1504.

\bibitem{shepherd1995prevention}
Shepherd J, Cobbe SM, Ford I, et al. Prevention of coronary heart disease with
  pravastatin in men with hypercholesterolemia. {\it New England Journal of
  Medicine} 1995\string; 333(20)\string: 1301--1308.

\bibitem{ferreira2007problems}
Ferreira-Gonz{\'a}lez I, Permanyer-Miralda G, Domingo-Salvany A, et al.
  Problems with use of composite end points in cardiovascular trials:
  systematic review of randomised controlled trials. {\it Bmj} 2007\string;
  334(7597)\string: 786.

\bibitem{tomlinson2010composite}
Tomlinson G, Detsky AS. Composite end points in randomized trials: there is no
  free lunch. {\it Jama} 2010\string; 303(3)\string: 267--268.

\bibitem{abdelhady2021effect}
Abdelhady H, Abdelrazik M, Abdi Z, et al. Effect of Convalescent Plasma on
  Organ Support--Free Days in Critically Ill Patients With COVID-19: A
  Randomized Clinical Trial. {\it Jama} 2021\string; 326(17)\string:
  1690--1702.

\bibitem{geman1984stochastic}
Geman S, Geman D. Stochastic relaxation, Gibbs distributions, and the Bayesian
  restoration of images. {\it IEEE Transactions on pattern analysis and machine
  intelligence} 1984(6)\string: 721--741.

\bibitem{pocock1977group}
Pocock SJ. Group sequential methods in the design and analysis of clinical
  trials. {\it Biometrika} 1977\string; 64(2)\string: 191--199.

\bibitem{o1979multiple}
O'Brien PC, Fleming TR. A multiple testing procedure for clinical trials. {\it
  Biometrics} 1979\string: 549--556.

\bibitem{gordon1983discrete}
Gordon~Lan K, DeMets DL. Discrete sequential boundaries for clinical trials.
  {\it Biometrika} 1983\string; 70(3)\string: 659--663.

\bibitem{zhu2017bayesian}
Zhu H, Yu Q. A Bayesian sequential design using alpha spending function to
  control type I error. {\it Statistical methods in medical research}
  2017\string; 26(5)\string: 2184--2196.

\bibitem{thompson1933likelihood}
Thompson WR. On the likelihood that one unknown probability exceeds another in
  view of the evidence of two samples. {\it Biometrika} 1933\string;
  25(3-4)\string: 285--294.

\bibitem{trippa2012bayesian}
Trippa L, Lee EQ, Wen PY, et al. Bayesian adaptive randomized trial design for
  patients with recurrent glioblastoma. {\it Journal of Clinical Oncology}
  2012\string; 30(26)\string: 3258.

\bibitem{annane2018hydrocortisone}
Annane D, Renault A, Brun-Buisson C, et al. Hydrocortisone plus fludrocortisone
  for adults with septic shock. {\it New England Journal of Medicine}
  2018\string; 378(9)\string: 809--818.

\bibitem{pocock2012win}
Pocock SJ, Ariti CA, Collier TJ, Wang D. The win ratio: a new approach to the
  analysis of composite endpoints in clinical trials based on clinical
  priorities. {\it European heart journal} 2012\string; 33(2)\string: 176--182.

\bibitem{doob1953stochastic}
Doob JL, Doob JL. {\it Stochastic processes}. 7.
\newblock Wiley New York .
\newblock 1953.

\bibitem{van2000asymptotic}
Vaart V.~dAW. {\it Asymptotic statistics}. 3.
\newblock Cambridge university press .
\newblock 2000.

\end{thebibliography}

\end{document}